%% file: main.tex
\documentclass[11pt]{article}
\usepackage{fullpage}
\usepackage{geometry}
\input{macros}

\title{Robust Max Selection}
\author{Trung Dang  \\ {\tt dddtrung@cs.utexas.edu}
\and
Zhiyi Huang \\ {\tt zhiyih@cs.utexas.edu}\footnote{All co-authors are affiliated with The University of Texas at Austin.}
}

\date{}

\begin{document}

\maketitle

\thispagestyle{empty}
\addtocounter{page}{-1}

\begin{abstract}
    We introduce a new model to study algorithm design under unreliable information, and apply this model for the problem of finding the uncorrupted maximum element of a list containing $n$ elements, among which are $k$ corrupted elements. Under our model, algorithms can perform black-box comparison queries between any pair of elements. However, queries regarding corrupted elements may have arbitrary output. In particular, corrupted elements do not need to behave as any consistent values, and may introduce cycles in the elements' ordering. This imposes new challenges for designing correct algorithms under this setting. For example, one cannot simply output a single element, as it is impossible to distinguish elements of a list containing one corrupted and one uncorrupted element. To ensure correctness, algorithms under this setting must output a set to make sure the uncorrupted maximum element is included.

    We first show that any algorithm must output a set of size at least $\min\{n, 2k + 1\}$ to ensure that the uncorrupted maximum is contained in the output set. Restricted to algorithms whose output size is exactly $\min\{n, 2k + 1\}$, for deterministic algorithms, we show matching upper and lower bounds of $\Theta(nk)$ comparison queries to produce a set of elements that contains the uncorrupted maximum. On the randomized side, we propose a 2-stage algorithm that, with high probability, uses $O(n + k \polylog k)$ comparison queries to find such a set, almost matching the $\Omega(n)$ queries necessary for any randomized algorithm to obtain a constant probability of being correct.
\end{abstract}

\newpage
\input{sec_intro}
\input{sec_prelims}
\input{sec_deterministic}
\input{sec_random}
\input{sec_conclusion}

\newpage
\bibliographystyle{plainnat}
\bibliography{main}


\end{document}

%% file: macros.tex
\usepackage{xcolor}
\usepackage{dirtytalk}
\usepackage{amsfonts}
\usepackage{amsthm}
\usepackage{amsmath}
\usepackage{amssymb}
\usepackage{mathtools}
\usepackage{bbm}
\usepackage[colorlinks,citecolor=blue]{hyperref}
\usepackage{cleveref}
\usepackage{enumerate}
\usepackage{comment}
\usepackage{algorithm}
\usepackage[noend]{algpseudocode}

\usepackage[numbers,sort]{natbib} 
\usepackage{thmtools}
\usepackage{thm-restate}
\renewcommand{\cite}[1]{\citep{#1}}
\newcommand{\textcite}[1]{\citet{#1}}

\newtheorem{theorem}{Theorem}

\newtheorem{lemma}[theorem]{Lemma}

\theoremstyle{definition}

\newcommand{\todo}[1]{}

\newcommand{\edit}[1]{}

\newcommand{\E}{\mathbb{E}}

\algnewcommand{\LineComment}[1]{\State \hfill \(\triangleright\) #1}

\newcommand{\polylog}{\operatorname{polylog}}

\newcommand{\DetMaxFind}{\textsc{DetMaxFind}}
\newcommand{\PruneAndRank}{\textsc{PruneAndRank}}

%% file: sec_intro.tex
\section{Introduction} \label{sec:intro}

One of the most mundane subroutines of many algorithms is finding the maximum of a list of elements. Unfortunately, the naive approach of passing through the list once is not robust to disturbances in the input elements, as a single corrupted element can alter its output drastically. In highly distributed processes where input data are owned by multiple actors, many of whom may act adversarially, the highly sensitive nature of the naive algorithm is simply unacceptable. 
For example, an actor controlling a small-valued element may act as if their element is larger than all other elements, thus the naive algorithm will output this small-valued element as the maximum instead of the true one.
It is therefore imperative to design an alternative approach that is robust under adversarial input corruption, while sacrificing little to no cost in performance.

Indeed, many previous works have investigated algorithm design under unreliable information. One of the most natural ways to model unreliable information is the \emph{liar model}, where algorithms can ask a (possibly lying) adversary for pairwise comparison queries between elements. Another model that has been extensively studied is the \emph{faulty-memory model}, which allows arbitrary memory faults during the course of the algorithm. However, unreliable information under these models may still cooperate well with the algorithm's execution, while that might not be the case with adversarial input corruption. To give an example, under the faulty-memory model, corrupted memory words still hold actual values in them, and therefore comparison queries regarding these memory words are still consistent with the, albeit corrupted, contained values. This assumption may be too strong in real life, where actors owning corrupted input elements can try to wreak havoc on the algorithm by answering queries in a way that is inconsistent with any value.

To amend the inadequacy of previous models, in this work, we introduce a new model that captures such adversarial corruption from the input data. At a high level, algorithms in our model can interact with input elements by querying a pairwise comparison black box. However, queries regarding corrupted elements may have arbitrary output, and in particular corrupted elements \emph{do not} have to behave as if they are any consistent values. For example, a corrupted element may answer ``no'' on the question ``is this element larger than $2$?'', but answer ``yes'' on the question ``is this element larger than $3$?''. This arbitrary behavior introduces new challenges in designing correct algorithms under this setting. As a simple example, for the task of finding the uncorrupted maximum element, an algorithm that is given one corrupted and one uncorrupted element is unable to pinpoint which element should be the uncorrupted maximum, as the comparison between the two elements is completely uninformative towards the goal.

For an introductory application of this model, we give preliminary results on the problem of finding the uncorrupted maximum. To ensure correctness, we allow algorithms under this model to output \emph{a set of elements} that must contain the uncorrupted maximum. We interpret this as giving a list of candidates for the central party that wants to know the maximum, using which this party can manually inspect the elements, for instance by contacting the actor owning each candidate element. As manual inspection is much more costly than asking pairwise comparisons, the natural objective for any algorithm is to first minimize the candidate set, and then to minimize the number of black-box pairwise comparisons used.

Our contributions are summarized as follows.
\begin{enumerate}
    \item In~\Cref{sec:alg-constraint}, we first show that any algorithm must produce a set of at least $\min\{n, 2k + 1\}$ elements to ensure that the maximum element is included, where $n$ and $k$ are the total number of elements and the number of corrupted elements respectively (\Cref{lem:setsizelb,lem:setsizeub}). Therefore, our work focuses on algorithms that produce a set of exactly such size.
    \item On the deterministic side, we show that no algorithms can use fewer than $(1 - o(1))nk$ queries (\Cref{t:detlb}), assuming $n \gg k$. On the upper bound side, we design an algorithm $\DetMaxFind$ that uses $(2 + o(1))nk$ queries (\Cref{t:det}).
    \item On the randomized side, we show that no algorithms can use fewer than $\Omega(n)$ queries to produce a set containing the maximum with $3/4$ probability (\Cref{thm:randomized-lb}). To complement this, we provide a two-stage algorithm $\PruneAndRank$ such that, with $1 - O\left(\frac{1}{\polylog k}\right)$ probability, uses $O(n + k \polylog k)$ queries and contains the maximum in its output (\Cref{thm:main-thm-randomized}). We remark that these bounds are matching under the assumption that $n \gg k$.
\end{enumerate}

\subsection{Related works}

Algorithm design under unreliable information is a classic topic in theoretical computer science, started by the work of~\cite{von1956probabilistic}. Various settings and models have been introduced to capture different ways information can become unreliable. We give a non-comprehensive list of studied settings most relevant to our work below.

\paragraph{Liar model.} Under this model, algorithms can ask queries to a possibly lying adversary. Different types of assumptions can be made about the adversary; for example, \emph{constant bounded errors} where up to some number of queries can be incorrect~\cite{rivest1980coping,ravikumar1987selecting,ravikumar2002fault}, \emph{probabilistic errors} where queries can be incorrect with some probability independently~\cite{berry1986discrete,feige1994computing,braverman2016parallel,gu2023optimal,gal1978stochastic,horstein1963sequential,pelc2002searching}, or \emph{prefix-bounded errors} where any prefix of $i$ queries have at most $pi$ incorrect queries~\cite{aslam1991searching,borgstrom1993comparison,pelc2002searching}. The most relevant work under this model to our work is~\cite{ravikumar1987selecting}, where they showed that it is possible to find the maximum element using $(k + 1)n - 1$ comparisons if at most $k$ queries are lies. For the closely related problem of sorting,~\cite{ravikumar2002fault} shows that it is possible to sort an array in $O(n \log n)$ under the constant bounded setting if there are at most $O(\frac{\log n}{\log \log n})$ faults; ~\cite{gu2023optimal} gives the matching bounds to sorting and binary search under probabilistic errors setting up to $1 + o(1)$ factor, improving upon~\cite{feige1994computing,wang2022noisy}; finally,~\cite{borgstrom1993comparison} shows that under the prefix-bounded setting, it takes an exponential amount of queries just to determine whether a list is sorted. We note that as the corruption in this model chiefly lies in the queries and not the input data, this model can take advantage of repeating the same queries multiple times, whereas it is not possible to do the same in our case.

\paragraph{Faulty memory model.} Introduced by~\cite{finocchi2004sorting}, this model concerns with \emph{memory failure} that can happen anywhere and at any place during the execution of the algorithm, except for $O(1)$ memory words that are reliable. Since then, a number of different works~\cite{ferraro2006price,finocchi2008sorting,finocchi2009optimal,kopelowitz2012selection,gieseke2012resilient,caminiti2017resilient,leucci2019resilient} have investigated different algorithmic tasks under this model. For the task of selection,~\cite{kopelowitz2012selection} designed a deterministic $O(n)$ selection algorithm that, on for a given $k$ and number of memory faults $\alpha$, return some \emph{possibly faulty} elements with rank between $k - \alpha$ and $k + \alpha$. We remark that while this result is very similar to our result superficially, since corrupted data in this model \emph{is still actual data}, it is fundamentally different from our result as our corrupted elements may not even behave consistently to any real data. For the related task of sorting,~\cite{finocchi2004sorting} designed an $O(n \log n + \alpha^3)$ algorithm to sort uncorrupted data, and showed that any $O(n \log n)$ algorithm can only tolerate up to $O((n \log n)^{1/2})$ memory faults. These bounds are later tightened by~\cite{finocchi2008sorting}, who showed an $O(n \log n + \alpha^2)$ algorithm for the same task.

\paragraph{Fault-Tolerant Sorting Networks.}
This model mainly focuses on the concept of destructive faults in fault-tolerant sorting networks \cite{assaf1990fault,leighton1999tight,366814,yao1985fault}, where the comparators in the sorting network could be faulty, and can even destroy one of the inputs by outputting one same value twice. When those destructive faults are random, \cite{assaf1990fault} presents a sorting network with size $O(n\log^2{n})$ and completes the sorting task with high probability. A following work by \cite{leighton1999tight} shows that this bound is tight. The main technique of this fault-tolerant sorting network is duplicating input values using fault-free data replicators, which we once again remark is impossible to do in our case.

%% file: sec_prelims.tex
\section{Preliminaries}

\subsection{Model}

Formally, our model is as follows. There are $n$ elements labeled $X = \{x_1, x_2, \dots, x_n\}$, where any two elements can be compared. In other words, there exists a complete directed graph $G$ on the $n$ elements, which we will call the \emph{comparison graph}. For any two elements $x_u$ and $x_v$, we say $x_u$ is \emph{larger} than $x_v$ if the edge between them is directed towards $u$, and $x_u$ is smaller than $x_v$ vice versa. Among these $n$ elements are $k$ special \emph{corrupted} elements, and the remaining $n - k$ elements are \emph{uncorrupted}. It is known that the edges between the $n - k$ uncorrupted elements form a complete directed acyclic graph, while edges incident to the $k$ corrupted elements may be arbitrarily, possibly adversarily directed. In particular, there may be directed cycles that go through any of the directed elements in the comparison graph.

An algorithm is given $n$ and $k$, as well as a black box that can access the direction of the edge between any pair of elements, called the \emph{pairwise comparison black box}. The algorithm is then allowed to repeatedly query for the edge direction between any pair of elements. Finally, it must output \textbf{a set of elements} that contains the \emph{uncorrupted maximum} (sometimes shortened to just the \emph{maximum}), i.e. some uncorrupted element $x_m$ such that for every other uncorrupted element $x_i$, $x_m$ is larger than $x_i$. We note that such an element always exists, as the induced subgraph of the comparison graph on uncorrupted elements forms a complete DAG as discussed before. The algorithm's objective is to find such a set using as few black box comparisons as possible.

Under the model so far, an algorithm can simply output the set of all elements, which trivially contains the uncorrupted maximum. Therefore, we will restrict our attention to algorithms whose output size is as small as possible before optimizing for the number of pairwise comparisons. In the next section, we discuss why the output size of such an algorithm should be $\min\{n,2k + 1\}$.

\subsection{Algorithmic constraint}\label{sec:alg-constraint}

Under the presented model, we remark that it is not possible to output a set containing one element only. As a simple example, consider a graph of $n = 2$ elements, where $k = 1$ element is uncorrupted. As the direction of the edge between the two elements is completely uninformative toward determining the uncorrupted element, any algorithm must output both elements to ensure correctness.

This poses an essential question for this model: for the set $S$ that any algorithm produces, how large does $S$ need to be such that the maximum is always contained in $S$? We answer the question via the following lemma.

\begin{lemma}
\label{lem:setsizelb}
For any $n$ and $k$, there exists an instance where the output set size $|S|$ for any algorithm is at least $\min\{n,2k+1\}$ to ensure that the maximum element is always included.

\end{lemma}

\begin{proof}[Proof of~\Cref{lem:setsizelb}]

The main idea of the proof is constructing a special example with $n = 2k+1$ elements, where $k$ of them are corrupted. Each element looks symmetric in the comparison graph.

Explicitly, consider those $2k+1$ elements $\{x_1,x_2,\dots,x_{2k+1}\}$ where $x_{k+1} \to x_k \to \ldots \to x_1$ are uncorrupted elements with such topological order so the uncorrupted maximum is $x_1$, and elements $\{x_{k+2},x_{k+3},\dots,x_{2k+1}\}$ are corrupted. Observe we can carefully design the edges incident to the corrupted elements such that for any $i\in [2k+1]$, $x_i$ is larger than $\{x_{i+1},x_{i+2},\dots,x_{i+k}\}$ and smaller than $\{x_{i-1},x_{i-2},\dots,x_{i-k}\}$, where all the indices are with respect to modulo $2k+1$; notice that all the comparison results of uncorrupted elements are consistent to the aforementioned topological order. It's easy to observe that every element inside $\{x_1,x_2,\dots,x_{2k+1}\}$ is symmetric according to the comparison results, so no algorithms can distinguish between them and therefore must output the whole set.

When $n >  2k+1$, we can embed the example above with $\{x_1,x_2,..,x_{k+1}\}$ being the largest $k+1$ elements among the $n - k$ uncorrupted elements, while $\{x_{k+2},\dots,x_{2k+1}\}$ are the $k$ corrupted elements; we further let every element in this set be larger than the remaining elements. Since every element in this set of size $2k+1$ behaves symmetrically, any algorithm must output the full set to ensure the maximum is included.

When $n<2k+1$, we can construct a similar example where $x_{n - k} \to x_{n - k - 1} \to \ldots \to x_1$ are uncorrupted elements with such topological order so $x_1$ is the uncorrupted maximum, and elements $\{x_{n - k + 1},x_{n - k + 2}, \dots, x_n\}$ are corrupted. Then, we direct all of the edges such that for all $i\in [n]$, $x_i$ is greater than $\{x_{i+1},x_{i+2},\dots,x_{i+\lfloor (n - 1) / 2 \rfloor}\}$ and smaller than $\{x_{i-1},x_{i-2},\dots,x_{i-\lceil (n - 1) / 2 \rceil}\}$. It is easy to verify that no algorithms can distinguish between any element and therefore must output the full set.
\end{proof}

A natural follow-up is to determine whether an algorithm that outputs a set of size $\min\{n, 2k+1\}$ exists. The following lemma answers that question positively.

\begin{lemma}\label{lem:setsizeub}
    For every $i \in [n]$, let $r_i$ be the number of other elements $x_j$ larger than $x_i$ (which can be calculated using $n - 1$ queries). Let $S$ be the set of $\min\{n, 2k+1\}$ elements with the smallest $r_i$'s (ties are broken arbitrarily). Then $S$ contains the uncorrupted maximum.
\end{lemma}

\begin{proof}[Proof of~\Cref{lem:setsizeub}]
    Let $m \in [n]$ be the index of the uncorrupted maximum, and note that $r_m \le k$. We further observe that there are at most $\min\{n, 2k+1\}$ elements whose $r$ values are at most $k$, namely the $k$ corrupted elements and the top $\min\{k + 1, n - k\}$ uncorrupted elements. Therefore, the set $S$ from the lemma must contain the maximum element.
\end{proof}

Therefore, for the rest of this work, we assume that $n > 2k + 1$ (as one can simply output all elements otherwise) and consider algorithms whose output size is exactly $2k + 1$.

%% file: sec_deterministic.tex
\section{Deterministic Algorithm} \label{sec:deterministic}


In this section, we will show a lower bound of $\Omega(nk)$ queries on deterministic algorithms. We will also describe our deterministic algorithm that outputs a set of size $2k+1$ which contains the maximum element, and uses $O(nk)$ comparison queries.

For the lower bound, we first show the following theorem.

\subsection{Lower bound}

We first show that any deterministic algorithm for this setting must use $\Omega(nk)$ comparison queries.

\begin{theorem}
\label{t:detlb}
Any deterministic algorithm requires $(n-(2k+1))(k+1)=nk+(n-2k^2-3k-1)$ queries to ensure the maximum is contained in the output set.
\end{theorem}

\begin{proof}[Proof of~\Cref{t:detlb}]
At a high level, we construct two instances such that any algorithm that uses fewer queries than the stated quantity behaves identically under the two instances. However, one of the instances has its maximum outside of the output set of the algorithm.

For the first instance, we construct its comparison graph $G$ where for all $1 \le i < j \le n$, we let $x_j$ be larger than $x_i$, i.e. the edge goes $x_i \to x_j$. We remark that any $k$ elements in this instance can be the corrupted elements. Therefore, we will not specify the $k$ uncorrupted elements now, and we will use this as our leverage point later.

Consider any deterministic algorithm that performs strictly fewer than $(n-(2k+1))(k+1)$ queries on this instance, and let $S$ be the output set of the algorithm, where we recall that $|S| = 2k + 1$. Since every comparison query adds one to a ``smaller'' counter of exactly one element, by pigeonhole principle, there exists an element $x_m \notin S$ where $x_m$ was compared to be smaller than at most $k$ other elements in $X$. Let the set of at most $k$ elements larger than $x_m$ be $C$. We now fix the first instance such that every element in $C$ (along with some other arbitrary elements to make the total $k$) is a corrupted element. Finally, albeit an unimportant detail, we let the uncorrupted maximum to be the element with the smallest index that is not corrupted.

We now construct the second instance as follows. We let the comparison graph $G'$ be the same as $G$. We then modify the edges incident to $x_m$ so that for every other element $x_i$, $x_m$ is smaller than $x_i$ if and only if $x_i \in C$. We then let the uncorrupted maximum be $x_m$, and keep the same set of corrupted elements the same as the first instance. Since only $x_m$ has its edges changed, and $x_m$ is only smaller than elements in $C$ which are defined to be corrupted, the second instance is valid.

We remark that the only difference between the comparison graphs of the two instances are all edges $x_m \to x_i$ in $G$, where $x_i$ is uncorrupted in both instances, were changed to be $x_m \leftarrow x_i$ in $G'$. Furthermore, we observe that all queries made by the algorithm in the first instance never ask about these edges. Otherwise, $x_i$ would have been in $C$, which contradicts the fact that $x_i$ is uncorrupted in the first instance. Therefore, the sequence of queries and their black-box output of the algorithm for both instances are completely identical, and therefore the output set $S$ for both instances is the same. However, under the second instance, $S$ does not contain $x_m$ by definition, which means the algorithm has failed to output the uncorrupted maximum in the second instance.
\end{proof}

In practice, we can assume that $n \gg k$, turning the result of \Cref{t:detlb} into a lower bound of $(1 - o(1)) nk$ for deterministic algorithms. 

\subsection{Upper bound}

To complement the previous result, we now show a deterministic algorithm that uses $O(nk)$ queries. Under the assumption that $n \gg k$, this bound is tightened to be $(2 + o(1)) nk$ queries.

\begin{theorem}
\label{t:det}
There exists a deterministic algorithm, which outputs a set $S$ with size $2k+1$ such that the maximum is contained in $S$, and uses $(n-(k+1))(2k+1)=2nk+(n-2k^2-3k-1)$ queries.
    
\end{theorem}

At a high level, our deterministic algorithm maintains a set of size $2k+1$, while making sure the maximum element is inside this set during the whole process. Formally, our algorithm is described in~\Cref{alg:det}.

\begin{algorithm}[h!]
\caption{$\DetMaxFind(X = \{x_1, x_2, \dots, x_n\})$}
\label{alg:det}
\begin{algorithmic}[1]
    \State Let $S$ be $\emptyset$ initially.
    \For {$i = 1 \to n$}
        \State Compare $x_i$ with all the elements in $S$ and cache the results of these pairwise comparisons.
        \State $S \leftarrow S\cup \{x_i\}$.
        \If {$|S|>2k+1$}
            \State Let $\Bar{x}$ be some element in $S$ which is smaller than at least $k+1$ other elements in $S$.
            \State $S \leftarrow S\setminus \{\Bar{x}\}$.
        \EndIf
    \EndFor
\State \Return $S$.
\end{algorithmic}
\end{algorithm}

Observe that when we enter line 6 of~\Cref{alg:det}, there are exactly $2k+2$ elements in $S$. Thus by pigeonhole principle, we can always find an element that is smaller than at least $k+1$ other elements in $S$. Furthermore, the algorithm outputs a set $S$ of size $2k+1$, which fulfills our restriction. Now we will prove the maximum is contained in this output set.

\begin{lemma}
\label{lem:detoutputmax}
The output set of \Cref{alg:det} contains the maximum element.
\end{lemma}

\begin{proof}[Proof of~\Cref{lem:detoutputmax}]
The maximum element can only be smaller than at most $k$ corrupted elements. In line 6 of \Cref{alg:det}, the element we remove from $S$ is at least smaller than $k+1$ other elements. Thus the maximum element will never be removed from $S$.
\end{proof}

Lastly, we will show the number of queries we make as claimed.

\begin{theorem}
\label{t:detquerynum}
\Cref{alg:det} makes $(n-(k+1))(2k+1)=2nk+(n-2k^2-3k-1)$ queries.
\end{theorem}

\begin{proof}[Proof of~\Cref{t:detquerynum}]
We only perform pairwise comparisons in line 3. When $i\leq 2k+1$, since the size of $S$ is $i-1$, we do $i-1$ queries. When $i\geq 2k+2$, we maintain a set $S$ with size $2k+1$, thus we do $2k+1$ each.

Additionally, line 3 ensures that at all times, we have the pairwise comparison results of all pairs of elements in $S$, so line 6 does not require any extra query as we can use the stored cache to figure out $\bar{x}$.

Therefore, the total number queries is $1+2+\ldots+2k+(2k+1)(n-(2k+1))=(n-(k+1))(2k+1)=2nk+(n-2k^2-3k-1)$.
\end{proof}

%% file: sec_random.tex
\section{Randomized Algorithm} \label{sec:random}

In this section, we provide a lower bound of $\Omega(n)$ and an upper bound of $O(n + k \polylog k)$, under some mild assumptions on the bound of $k$. These bounds are almost matching in practice due to the assumption of $n \gg k$.

\subsection{Lower bound}

We first show that under a mild assumption, any randomized algorithm that wants to achieve a vanishing probability of containing the maximum in its output set must use $\Omega(n)$ comparison queries. Intuitively, this makes sense as any algorithm should use so much queries to gain information about most of the elements in the first place.

\begin{theorem}\label{thm:randomized-lb}
    Assuming $2k + 1 \le \frac{n}{4}$, then no algorithms can output a set of $2k + 1$ elements containing the maximum with probability more than $3/4$ using $\frac{n}{4}$ comparison queries.    
\end{theorem}

The assumption above is relatively mild, as 1. in practice we should expect $n \gg k$, and 2. otherwise, one can just output a random subset of the elements and achieve a constant probability of including the maximum.

Let us first describe the hard instance for~\Cref{thm:randomized-lb}. We consider the instance from~\Cref{lem:setsizelb} where there is a cycle of $2k + 1$ elements of the top $k + 1$ uncorrupted and $k$ corrupted elements, then we shuffle the indices of all $n$ elements uniformly randomly. Let us call the $2k + 1$ elements in the cycle as \emph{critical elements}, denoted by the set $C$.

Without loss of generality, fix any randomized algorithm $\textsc{Alg}$ that uses exactly $\frac{n}{4}$ queries (as the algorithm can repeat queries as necessary). We remark that since no algorithms can distinguish elements in $C$, combined with the random indexing, the probability that any algorithm produces a set containing the maximum is exactly $\E\left[\frac{|S \cap C|}{|C|}\right] = \E\left[\frac{|S \cap C|}{2k + 1}\right]$ where $S$ is the output set of the algorithm. Our task is reduced to upper bounding $\E[|S \cap C|]$.

Our first technical lemma shows a relationship between $\E[|S \cap C|]$ and $\E[|T \cap C|]$, where $T$ is the random set of elements that are included in at least one pairwise comparison that $\textsc{Alg}$ made; note that $|T| \le \frac{n}{2}$.

\begin{lemma}\label{lem:tech-lem-1-randomized}
    $\E[|S \cap C|] \le \frac{\E[|C \cap T|]}{2} + \frac{2k + 1}{2}$.
\end{lemma}

\begin{proof}[Proof of~\Cref{lem:tech-lem-1-randomized}]
First note that
\begin{equation}\label{eq:lemma-9-eq}
    \E[|S \cap C|] = \E[|(S \cap C) \cap T|] + \E[|(S \cap C) \setminus T|] \le \E[|C \cap T|] + \E[|(S \cap C) \setminus T|]
\end{equation}

We focus on the quantity $(S \cap C) \setminus T$: this is the set of outputted critical elements that the algorithm does not observe. Conditioned on $T$, we note that elements $j \notin T$ have equal probability of being in $C$ since the algorithm cannot distinguish between these elements, and this probability is exactly $\frac{|C \setminus T|}{n - |T|}$. Therefore,
\[\E[(S \cap C) \setminus T] = \E\left[|S \setminus T| \cdot \frac{|C \setminus T|}{n - |T|}\right] \le \E\left[|C \setminus T| \cdot \frac{|S|}{n - n/2}\right] = \E\left[|C \setminus T| \cdot \frac{2k + 1}{n/2}\right] \le \frac{\E[|C \setminus T|]}{2}\]
where we use $2k + 1 \le \frac{n}{4}$ by assumption.

Plugging this back into~\Cref{eq:lemma-9-eq}, we have
\[\E[|S \cap C|] \le \E[|C \cap T|] + \frac{\E[|C \setminus T|]}{2} = \frac{\E[|C \cap T|]}{2} + \frac{\E[|C|]}{2} = \frac{\E[|C \cap T|]}{2} + \frac{2k + 1}{2}.\]
\end{proof}

The final step is to bound $\E[|C \cap T|]$. Intuitively, since $T$ is a subset of size at most $n/2$, and the instance is randomly shuffled, we should expect to see $\E[|C \cap T|] \le \frac{2k + 1}{2}$ no matter how the algorithm proceeds. The following lemma captures this intuition.

\begin{lemma}\label{lem:tech-lem-2-randomized}
$\E[|C \cap T|] \le \frac{2k + 1}{2}$.
\end{lemma}

\begin{proof}[Proof of~\Cref{lem:tech-lem-2-randomized}]
    Let $A$ be the set of all $n$ elements, and let $T_i$ be the set of elements seen from the first $i$ queries of $\textsc{Alg}$, where $0 \le i \le \frac{n}{4}$. Define the random variable $Y_i = \frac{|C \setminus T_i|}{|A \setminus T_i|}$, i.e. the fraction of critical elements among the unseen elements after step $i$. Our goal is to show that $\{Y_i\}_{i=0}^{n/4}$ forms a martingale.

    First, we will show that if $\{Y_i\}$ forms a martingale, then we have our desired result. As $\E[Y_0] = \frac{2k + 1}{n}$, by optional stopping theorem, we also have $\E[Y_{n/4}] = \frac{2k + 1}{n}$, or $\E\left[\frac{|C \setminus T|}{|A \setminus T|}\right] = \frac{2k + 1}{n}$. Combined with $\frac{|C|}{|A|} = \frac{2k + 1}{n}$, we have $\E\left[\frac{|C \cap T|}{|T|}\right] = \frac{2k + 1}{n}$, which implies the lemma statement as $|T| \le n/2$.

    Finally, we show that $\{Y_i\}$ is a martingale. Conditioned on $T_i$, before the $i+1$-th query, since the algorithm has not seen elements in $A \setminus T_{i + 1}$ as $T_{i + 1} \supseteq T_i$, every element in $A \setminus T_{i + 1}$ has equal probability of being critical. In other words,
    \[
        \E[Y_{i + 1} \mid Y_0, Y_1, \dots, Y_i] = \E\left[\frac{|C \setminus T_{i + 1|}}{|A \setminus T_{i+1}|} \;\Big|\; T_i\right] = \frac{|C \setminus T_i|}{|A \setminus T_i|} = Y_i.
    \]
\end{proof}

Combining the two lemma, we have $\E[|S \cap C|] \le \frac{1}{2} \cdot \frac{2k + 1}{2} + \frac{2k + 1}{2} = \frac{3}{4} \cdot (2k + 1)$, implying~\Cref{thm:randomized-lb}.

\subsection{Upper bound}

Finally, we describe our randomized algorithm that contains the maximum element in its output set with vanishing probability with $O(n + k \polylog k)$ queries, almost matching the lower bound from~\Cref{thm:randomized-lb}. The main theorem for this section is as follows.

\begin{theorem} \label{thm:main-thm-randomized}
    For any $0 < c \le 1$, there exists a randomized algorithm such that with probability at least $1 - O(k^{-c} \log k)$, the algorithm's output set contains the maximum using at most $O(n + k^{1 + 3c} \log k)$ queries.
\end{theorem}

By plugging in the appropriate value for $c$ (for example, a constant multiple of $\frac{\log \log k}{\log k}$), we can show that the algorithm uses at most $O(n + k \polylog k)$ queries to output the maximum with vanishing probability of $1 - O\left(\frac{1}{\polylog k}\right)$.

Our algorithm, parameterized by $c$ in~\Cref{thm:main-thm-randomized}, is characterized by two stages: the first stage prunes the number of elements in $X$ from $n$ to $k^{1+c}$ using $O(n)$ queries, and the second stage samples every element's rank and outputs a uniformly random subset among the top $2k + k^{1-c}$ ranking elements using $O(k^{1 + 3c} \log k)$ queries.

\begin{algorithm}[h!]
\caption{$\textsc{PruneAndRank}(X = \{x_1, x_2, \dots, x_n\}, c)$}
\label{alg:prune-and-rank}
\begin{algorithmic}[1]
    \LineComment{\emph{Stage 1:} Prune $|X|$ to $k^{1+c}$ elements}
    \State Let $t$ initially be $\varnothing$.
    \For {$k = 1 \to  \frac{2n \log k}{k^{1+c}}$}
        \State Let $x_j$ be a uniformly random element from $X$.
        \If {$t = \varnothing$ or $x_j > t$}
            \State $t \gets x_j$.
        \EndIf
    \EndFor
    \State Compare every element $x_i \in X$ against $t$, and remove those that are smaller than $t$.

    \LineComment{\emph{Stage 2:} Choose the top $2k + k^{1-c}$ sample ranks}
    \For {every element $x_i \in X$}
        \State Let $A_i$ be the a multiset of $3k^{2c} \log k$ elements from $X \setminus \{x_i\}$ chosen uniformly at random (with replacement).
        \State Compare $x_i$ against every element in $A_i$.
        \State Let $w_i$ be the number of elements larger than $i$ in $A_i$.
    \EndFor
    \State Let $T$ be $2k + k^{1-c}$ elements with the smallest $w_i$'s (ties are broken arbitrarily).
    \State Let $S$ be a uniformly random subset of $T$ of size $2k + 1$.
    \State \Return S.
\end{algorithmic}
\end{algorithm}

The rest of this section is dedicated to proving desirable properties of the two stages, from which~\Cref{thm:main-thm-randomized} follows immediately.

\paragraph{First Stage.}

Observe that the first stage uses at most $\frac{2n \log k}{k^{1+c}} + n \le 3n$ comparison queries. We now only need to show the following desirable property of stage 1.

\begin{lemma}
\label{lem:prune-and-rank-stage-1}
With probability at least $1 - O(k^{-c} \log k)$, the first stage reduces $X$ to a set of at most $k^{1+c}$ elements, among which is the maximum.
\end{lemma}

\begin{proof}[Proof of~\Cref{lem:prune-and-rank-stage-1}]
We observe that if every $\frac{2n \log k}{k^{1+c}}$ elements sampled from stage 1 is an uncorrupted element, then the set $S$ at the end of stage 1 contains the uncorrupted maximum. This is because $t$ at the end of the iteration is a uncorrupted element, and hence it cannot eliminate the uncorrupted maximum from $X$. The probability of this event happening is at least \[\left(1 - \frac{k}{n}\right)^{2n \log k/k^{1+c}} \ge 1 - 2k^{-c} \log k.\]

Conditioned on the event that all sampled elements are uncorrupted, we note that by the end, $t$ is the largest sampled element. We first calculate probability that $t$ at the end is among the largest $\frac{1}{2} k^{1+c}$ uncorrupted elements (conditioned on all sampled elements being uncorrupted); this probability is at least
\[1 - \left(1 - \frac{k^{1+c}}{2(n - k)}\right)^{2n \log k/k^{1+c}} \ge 1 - \left(1 - \frac{k^{1+c}}{2n}\right)^{2n \log k/k^{1+c}} \ge 1 - \left(\frac{1}{e}\right)^{\log k} = 1 - \frac{1}{k}.\]

This probability also lower bounds the probability that $X$ contains at most $k^{1+c}$ elements conditioned on all sampled elements being uncorrupted, since if $t$ is among the largest $\frac{1}{2} k^{1+c}$ uncorrupted elements, then the set $X$ after pruning would have at most $\frac{1}{2} k^{1+c} + k \le k^{1+c}$ elements. A simple union bound then verifies the theorem.
\end{proof}

\paragraph{Second Stage.}

In this stage, we use $X$ to refer to to the set $X$ at the beginning of stage 2, which is assumed to have at most $k^{1+c}$ elements, among which is the maximum. With this assumption, this stage uses at most $3k^{1+3c} \log k$ comparison queries. We are now left to show the following lemma.

\begin{lemma}\label{lem:prune-and-rank-stage-2}
    With probability at least $1 - \frac{1}{k}$, the set $T$ from line 13 contains the maximum element. Corollarily, the returned set contains the maximum element with probability $1 - O(k^{-c})$ by line 14 of~\Cref{alg:prune-and-rank}.
\end{lemma}

Towards proving this theorem, we let $r_i$ be the rank of $x_i$, i.e. the number of elements that is larger than $x_i$ in $X$. The following technical lemma ensures separability between the sampled ranks $w_i$'s if there was a large enough gap in the true ranks.

\begin{lemma}\label{lem:pairwise-comparison-high-prob}
    For any two elements $x_i, x_j \in X$, if $r_i - r_j \ge k^{1-c}$, then $w_i > w_j$ with probability at least $1 - \frac{1}{k^{2+c}}$.
\end{lemma}

\begin{proof}[Proof of~\Cref{lem:pairwise-comparison-high-prob}]
Let $q = 3k^{2c} \log k$. Observe that $w_j - w_i$ is distributed as $\sum_{k = 1}^q B_k - \sum_{k=1}^q C_k$, where $B_k$'s and $C_k$'s are mutually independent Bernoulli random variables such that $\Pr[B_k = 1] = \frac{r_j}{|S|}$ and $\Pr[C_k = 1] = \frac{r_i}{|S|}$ for all $k \in [q]$.

As $\E[w_j - w_i] = \frac{q(r_j - r_i)}{|S|} \le -\frac{q k^{1-c}}{k^{1+c}} = -\frac{q}{k^{2c}}$, by Hoeffding's bound we have
\begin{align*}
    \Pr[w_j - w_i \ge 0] &\le \Pr\left[(w_j - w_i) - \E[w_j - w_i] \ge \frac{q}{k^{2c}}\right] \\
        &\le \exp\left(-\frac{2\left(\frac{q}{k^{2c}}\right)^2}{2q}\right) \\
        &= \exp(-3 \log k) = \frac{1}{k^3} \le \frac{1}{k^{2+c}}
\end{align*}
which means $\Pr[w_i > w_j] \ge 1 - \frac{1}{k^{2+c}}$.
\end{proof}

With this lemma, we can now prove~\Cref{lem:prune-and-rank-stage-2}.

\begin{proof}[Proof of~\Cref{lem:prune-and-rank-stage-2}]
    Let $x_m$ be the maximum, and note that $r_m \le k - 1$. Therefore, there are at most $2k + k^{1 - c}$ other elements $x_j \in X$ such that $r_j \le r_m + k^{1 - c}$ (including the $k$ corrupted elements and the largest $k + k^{1+c}$ uncorrupted elements).
    
    For every element $x_o \in S$ that does not satisfy the previous condition, we know that $r_o - r_m > k^{1-c}$. By~\Cref{lem:pairwise-comparison-high-prob}, we have $w_m > w_o$ with probability at least $1-\frac{1}{k^{2+c}}$. By a union bound, this event holds true for all such $x_o$ with probability at least $1 - |X| \cdot \frac{1}{k^{2 + c}} \ge 1 - \frac{1}{k}$. In such an event, we know that $x_m$ will be among the $2k + k^{1-c}$ elements with the smallest $w_i$'s, and hence $T$ contains $x_m$.
\end{proof}

%% file: sec_conclusion.tex
\section{Conclusion and Open Problems}\label{sec:conclusion}

In this work, we have developed a novel model for algorithm design under unreliable information, specifically for the task of finding the uncorrupted maximum among a list containing possibly corrupted elements. We first settled the output size of any algorithm by showing that at least $\min\{n, 2k+1\}$ must be included for an algorithm to always contain the maximum element, where $k$ is the number of corrupted elements. Restricted to algorithms whose output size is exactly such quantity, we presented a deterministic and a randomized algorithm that uses almost matching number of queries to their respective lower bounds. For deterministic algorithms, we showed a lower bound of $\Omega(nk)$ queries, while our deterministic algorithm matched the lower bound by using only $O(nk)$ pairwise comparisons. For randomized algorithms, we presented a high probability randomized algorithm that uses only $O(n + k \polylog k)$ comparison queries, which almost matches the $\Omega(n)$ lower bound for any randomized algorithm.

As possible directions of future research, the most immediate task is to close the gaps for both deterministic and randomized algorithms. For deterministic algorithms, we remark that there is a constant gap of $2$, so it would be interesting if we could settle the exact number of pairwise comparisons. For randomized algorithms, we conjecture that there is also a $\Omega(k \polylog k)$ lower bound, due to the following argument. Consider the bad instance in~\Cref{thm:randomized-lb}. For deterministic algorithms, there is an easy argument for a $\Omega(k \log n)$ lower bound in this instance: since there can be $\binom{n}{2k + 1}$ ways to choose the critical set, any deterministic algorithm must use at least $\log \binom{n}{2k + 1} = \Omega(k \log n)$ queries to distinguish all instances (with the mild assumption that $2k + 1 \le n/2$). We believe this argument can also be extended to randomized algorithms, further closing the gap between the lower bound and the upper bound.

Another interesting question is whether we can do more complex tasks on this model of unreliable computation elements. A natural extension of this result would be the selection or partition problems, but a more interesting question is if we can do sorting or build data structures in this setting (for example, $k$-d trees can be made resilient in the faulty-memory model~\cite{gieseke2012resilient}).

\paragraph{Acknowledgements.}

We thank Thatchaphol Saranurak and Shuchi Chawla for helpful conversations during the prior and late stages of this work.